\documentclass[11pt,oneside]{amsart}
\usepackage{amsthm}
\usepackage{setspace}
\usepackage{amssymb}
\usepackage{esint}
\usepackage{mathrsfs}
\numberwithin{equation}{section} 
\numberwithin{figure}{section} 
\theoremstyle{plain}
\newtheorem{thm}{Theorem}[section]
  \theoremstyle{remark}
  
  \theoremstyle{plain}
  \newtheorem*{acknowledgement*}{Acknowledgement}
  \theoremstyle{plain}
  \newtheorem{lem}[thm]{Lemma}
  \theoremstyle{plain}
  \newtheorem{prop}[thm]{Proposition}
  \theoremstyle{definition}
  \newtheorem*{problem*}{Problem}
  \theoremstyle{plain}
  \newtheorem*{thm*}{Theorem}
  
  \theoremstyle{plain}

\begin{document}
\author{Jean-Ren\'{e} Chazottes}
\address{Centre de Physique Théorique\\ \'{E}cole Polytechnique\\ 91128 Palaiseau Cedex\\ France}
\email{jeanrene@cpht.polytechnique.fr}
\author{Michael Hochman}
\address{Department of Mathematics\\ Fine Hall, Washington Rd.\\ Princeton NJ 08540\\ USA.}
\email{hochman@math.princeton.edu}

\title{On the zero-temperature limit of Gibbs states}
\begin{abstract}
We exhibit Lipschitz (and hence H\"{o}lder) potentials on the full shift $\{0,1\}^{\mathbb{N}}$
such that  the associated Gibbs measures fail to converge as the temperature
goes to zero.
Thus there are ``exponentially decaying'' interactions
on the configuration space $\{0,1\}^{\mathbb Z}$ for which the
zero-temperature limit of the associated Gibbs measures
does not exist. 
In higher dimension, namely on the configuration space
$\{0,1\}^{\mathbb{Z}^{d}}$, $d\geq3$, we show that this
non-convergence behavior can occur for the equilibrium states of finite-range
interactions, that is, for locally constant potentials.
\end{abstract}
\maketitle

\section{\label{sec:Introduction}Introduction}

\subsection{Background}

The central problem in equilibrium statistical mechanics or thermodynamic
formalism is the description of families of Gibbs states for a given interaction.
Their members are parametrized by inverse temperature, magnetic field,
chemical potential, etc. The ultimate goal is then to describe the set
of Gibbs states as a function of these parameters. The zero temperature
limit is especially interesting since it is connected to ``ground states'', that
is, probability measures supported on configurations with minimal
specific energy \cite{vanEnterFernandezSokal93}.

The purpose of this article is to shed some light on the zero-temperature
limit in the case of classical lattice systems, that is, systems with
a configuration space of the form $F^{\mathbb Z^d}$, where $F$ is
a finite set. We consider shift invariant, summable interactions
$\Phi=(\Phi_B)_{B\subseteq\mathbb{Z}^{d},|B|<\infty}$.
For every $\beta>0$, we denote by $\mathcal{G}(\beta \Phi)$
the (nonempty) set of Gibbs states of $\Phi$ at inverse temperature
$\beta$.
It contains at least one shift-invariant Gibbs state \cite{georgii}.
The question we are interested in is:
\[
\textrm{What is the limiting behavior of}\;\, \mathcal{G}(\beta
\Phi)\;\textrm{as}\;\beta\to+\infty\,?
\]
When $\mathcal{G}(\beta \Phi)$ is a singletom, we denote
by $\mu_{\beta\Phi}$ its unique (and necessarily shift-invariant)
element. Then the previous question becomes:
\[
\textrm{Does the limit of}\;\,
\mu_{\beta\Phi}\,\,\textrm{exist},\;\textrm{as}\;\beta\to+\infty ?
\]
(Limits of measures should be understood in the weak-* sense).
For the class of interactions we consider, shift-invariant
Gibbs states are also equilibrium states. To define them,
we need to introduce the function
\[
\varphi(x):=\sum_{B\ni0}\frac{1}{|B|}\ \Phi_{B}(x).
\]
This function can be physically interpreted as the contribution of the 
lattice site $0$ to the energy in the
configuration $x$. (\footnote{Since the interaction is shift-invariant, we 
can take any lattice site. Other
definitions are possible \cite[section 3.2]{Ruelle04}, but all lead to the 
same expected value under a given shift-invariant measure.})

Equilibrium states at inverse temperature $\beta>0$ are then
shift-invariant measures which maximize the quantity
\begin{equation}
P_\beta(\nu) := \int \beta\varphi d\nu + h(\nu) 
\label{eq:variational-functional}
\end{equation}
over all shift-invariant probability measures $\nu$ on $F^{\mathbb Z^d}$. 
Here $h(\nu)$ is the Kolmogorov-Sinai entropy of $\nu$, and the supremum 
is called the (topological) pressure. When the lattice is one dimensional 
and the potential $\varphi$ is H\"{o}lder, there is a unique Gibbs measure 
which is also the unique equilibrium measure. For $d>1$ the notions 
generally are not equivalent and there may be multiple Gibbs states, even 
for finite-range
interactions, but any shift-invariant Gibbs measure is an equilibrium 
state.

As in \cite{vanEnterFernandezSokal93}, we define
zero-temperature equilibrium states as those shift-invariant
probability measures which maximize \footnote{Or minimizes, depending on 
the sign convention for $\varphi$.}
$\int \varphi d\nu$ among all shift-invariant measures $\nu$.
It can be proven that the weak-$^*$ accumulation points
of equilibrium states for a given interaction as $\beta\to+\infty$
are necessarily zero-temperature equilibrium states for that
interaction. Zero-temperature equilibrium states are related to
ground states (see \cite{vanEnterFernandezSokal93} for details).

\subsection{The one-dimensional case}

Let us make a few remarks about the ergodic perspective. Fix the usual 
metric
\[
d(x,y)=2^{-\max \{k\;:\; x_i= y_i\ \forall |i|\leq k\}}
\]
on $F^\mathbb{Z}$. For a number of reasons, the usual class of 
``potentials'' $\varphi:F^\mathbb{Z}\to\mathbb{R}$ which are studied are 
H\"{o}lder continuous ones. First, for these potentials the Gibbs measure 
$\mu_{\beta\varphi}$
is unique for each $\beta>0$ (no phase transition). Second, this class of 
potentials arises naturally in the theory of differentiable dynamical 
systems (e.g. Axiom A diffeomorphisms): By choosing a suitable Markov 
partition of the phase space one can code such a diffeomorphism to a 
subshift of finite type in $F^\mathbb{Z}$ \cite{bowenbook}, and under this 
coding smooth potentials lift to H\"{o}lder ones. And third, H\"{o}lder 
potentials correspond to the natural objects in statistical mechanics, 
namely ``exponentially decaying''
interactions $(\Phi_B)$ \cite[chapter 5]{Ruelle04}. We also note that the 
case when $\varphi$ is locally constant corresponds to interactions of 
finite range; see below.

There is a trick, due to Sinai, which allows one to reduce the study to a ``one-sided'' subshift of finite type
of $F^{\mathbb N}$ and a potential $\varphi$ which depends only on ``future'' 
coordinates \cite{bowenbook}. Thus it suffices to study the one-sided full shift, and our question can be formulated as follows: \[
\textrm{For H\"{o}lder continuous }\varphi\textrm{ on $F^\mathbb{N}$, when does }\lim_{\beta\to+\infty}\mu_{\beta\varphi}\textrm{ exist?}
\]

The existence of the zero-temperature limit has been verified in a number of situations, but, surprisingly, a systematic study of this question began only recently. When $d=1$ and $\varphi$ is locally constant (i.e. the interaction is finite range), the zero-temperature limit was proved to exist in \cite{Bremont03} and was described explicitly in \cite{Leplaideur05,ChazottesGambaudoUgalde09}. In this case, the zero-temperature limit is supported on the union of finitely many transitive subshifts of finite type and is a convex combination of the entropy-maximizing measure on them. The case $d=1$ with $F$ a countable set was studied in \cite{JenkinsonMauldinUrbanski05,morris}. 

Another class of examples where convergence may be verified arises as follows. Let $X\subseteq F^{\mathbb{N}}$ be a subshift (a closed non-empty shift-invariant set) and define $\varphi=\varphi_{X}$ by \[
\varphi(y)=-d(y,X)=-\inf\{d(y,x)\,:\, x\in X\}.
\]
This is a Lipschitz function on $F^{\mathbb{N}}$ with $\varphi|_{X}=0$ and $\varphi\leq0$. The ground states of $\varphi_{X}$ are then precisely the measures supported on $X$, and it follows that all accumulation points of
$\{\mu_{\beta\varphi}, \beta>0\}$ are invariant measures supported on $X$. In particular, when $X$ has only one invariant measure $\mu$ (i.e. is uniquely ergodic), all accumulation points coincide, and we have $\mu_{\beta\varphi}\rightarrow\mu$ as $\beta\rightarrow+\infty$.

The only example of non-convergence of which we are aware is by van Enter and Ruszel \cite{vanEnterRuszel07}. The example is of a nearest-neighbor potential model, but is defined over a continuous state space $F$ (the circle). 

This state of affairs has led to the belief that over finite state spaces convergence should generally hold. Our first result is a counterexample, showing that this is not the case:

\begin{thm}
\label{thm:main-theorem}
There exist subshifts $X\subseteq\{0,1\}^{\mathbb{N}}$
so that, for the Lipschitz potential $\varphi_{X}(y)=-d(y,X)$, the
family $\{\mu_{\beta\varphi}, \beta>0\}$ does not converge (weak-{*}) as $\beta\rightarrow+\infty$. 
\end{thm}

This theorem holds more generally for one-sided or two-sided mixing
shifts of finite type. 

Our construction gives reasonable control over the dynamics of $X$
and of the dynamics, number and geometry of the limit measures. An
interesting consequence of the construction is that the set of limit
measures need not be convex. We discuss these issues in section \ref{sec:Remarks}.

\subsection{The multi-dimensional case}

Our second result concerns higher dimension: non-convergence
can also arise when $d\geq 3$, {\em even for finite-range interactions};
contrast this with the positive result in dimension $d=1$ where
the zero temperature limit is known to exist in this case
\cite{Bremont03,Leplaideur05,ChazottesGambaudoUgalde09}.

While the methods used in the one-dimensional case are fairly classical and quite well-known
in the dynamics community, the study of zero-temperature limits
and ground states in higher dimensions turns out to be closely connected to symbolic dynamics,
and our results rely heavily on recent progress in understanding of multidimensional subshifts
of finite type, where computation theory plays a prominent role. Recall that a shift of finite type
$X\subseteq\{0,1\}^{\mathbb{Z}^{3}}$ is a subshift defined by a finite set $L$ of patterns and the condition
that $x\in X$ if and only if no pattern from $L$ appears in $x$.
Given $L\subseteq\{0,1\}^{E}$ one can define the finite-range interaction
$(\Phi_{B})_{B\subseteq\mathbb{Z}^{d},|B|<\infty}$ by
\[
\Phi_{E}(x)=\left\{ \begin{array}{cc}
-1/|E| & x|_{E}\in L\\
0 & \mbox{otherwise}\end{array}\right.
\]
and $\Phi_{B}=0$ for $B\neq E$; the associated potential on $\{0,1\}^{\mathbb{Z}^{d}}$
is
\[
\varphi_{L}(x):=\sum_{B\ni0}\frac{1}{|B|}\ \Phi_{B}(x)=\left\{ \begin{array}{cc}
-1 & x|_{E}\in L\\
0 & \mbox{otherwise}.
\end{array}\right.
\]
Clearly an invariant measure $\mu$ on $\{0,1\}^{\mathbb{Z}^{d}}$
satisfies $\int\varphi_{L}d\mu=0$ if and only if $\mu$ is supported
on $X$; thus the shift-invariant ground states are precisely the
shift invariant measures on $X$. In this sense $\varphi_{L}$
is similar to $\varphi_{X}$ (although there are some delicate differences,
as we shall see in section \ref{sec:The-multidimensional-case}). 

The main result of \cite{Hochman09} provides a general method for
transferring one-dimensional constructions to higher-dimensional SFTs
with corresponding directional dynamics. Using this we are able to
adapt the construction from theorem \ref{thm:main-theorem} to the multidimensional case
with a finite-range potential. 
Recall that an equilibrium state for a potential $\varphi$ at inverse temperature
$\beta$ is a shift-invariant measure which maximizes the functional $P_\beta$
of Eq. \eqref{eq:variational-functional}.

\begin{thm}
\label{thm:multidim}
For $d\geq 3$ there exist locally constant (i.e. finite-range) potentials $\varphi$
on $\{0,1\}^{{\mathbb Z}^d}$ such that for any family $(\mu_{\beta\varphi})_{\beta>0}$ in which
$\mu_{\beta\varphi}$ is an equilibrium state (i.e. a shift-invariant Gibbs state), the limit
$\lim_{\beta\to+\infty} \mu_{\beta\varphi}$ does not exist.
\end{thm}

Some comments are in order because the previous statement is rather subtle.
If there were a unique Gibbs state for each $\beta$ then there would be a unique
choice for $\mu_{\beta\varphi}$, and the previous result could be formulated
more transparently: there exist locally constant potentials such that 
$\lim_{\beta\to+\infty} \mu_{\beta\varphi}$ does not exist. But we believe
that in our example uniqueness does not hold at low temperatures.\newline
A more precise way to state the previous theorem is to say
that the set-valued sequence $(\mathcal{G}(\beta \Phi))_{\beta>0}$ does
not converge in Hausdorff metric topology.\newline
Our result is about continuous families and does not contradict
the fact that for each given family $(\mu_{\beta\varphi})_{\beta>0}$
of equilibrium states, there always exists a subsequence $(\beta_i)_{i\in\mathbb{N}}$
such that the limit $\lim_{i\to\infty} \mu_{\beta_i \varphi}$ exists. This is due
to compactness of the space of probability measures.\newline
There is nothing new in the fact that one can choose {\em some} divergent
family $\beta\mapsto \mu_{\beta\varphi}$ of equilibrium states. Think {\em e.g.}
of the Ising model below the critical temperature ($\beta$ large enough): One can choose a family which alternates between the $+$ and $-$ phases. However it is also possible to choose families which converge
to one of the ground states. Let us insist that in contrast to this kind of situation we prove
the existence of examples where it is {\em not possible} to choose {\em any} family which converges
to a ground state.\newline
Let us say a few words about the limitations of this result. First, it seems likely that our examples support
non-shift-invariant Gibbs states, i.e. Gibbs states which are not equilibrium states, and, furthermore,
we do not know if the statement extends to them. Hence the requirement of
shift-invariance. As for the restriction $d\geq 3$, the method used in our construction, which produces
a potential of the form $\varphi_{L}$ above, relies on the results from \cite{Hochman09}
which at present are not available in $d=2$; but probably they hold in that case as well.

\begin{problem*} 
For $d\geq2$, do there exist finite-range potentials on the 
$d$-lattice such that every family of Gibbs states 
$\{\mu_{\beta\Phi},\beta >0\}$ fails to converge as 
$\beta\to+\infty$? 
\end{problem*}

In the next section we construct the subshift $X$ of theorem \ref{thm:main-theorem}.
Section \ref{sec:Analysis} contains the analysis and proof of theorem
\ref{thm:main-theorem}. Section \ref{sec:Remarks} contains some
remarks and problems. Section \ref{sec:The-multidimensional-case}
discusses the multidimensional case.

\begin{acknowledgement*}
We are grateful to Mar\'{i}a Isabel Cortez for pointing out a gap
in an early version of this paper. We are also grateful to A. C. D. van Enter
for useful comments.
\end{acknowledgement*}

\section{\label{sec:Construction-of-X}Construction of $X$}

For each $k\geq0$ we define by induction integers $\ell_{k}$, and
finite sets of blocks $A_{k},B_{k}\subseteq\{0,1\}^{\ell_{k}}$. The
construction uses an auxiliary sequence of integers $N_{1},N_{2},\ldots$,
with $N_{1}\ldots N_{k}$ determining $A_{i},B_{i},\ell_{i}$ for
$i\leq k$. Here we treat the $N_{k}$ as given, but in fact at each
stage we are free to choose $N_{k+1}$ based on the construction so
far, and during the analysis in the next section we impose conditions
on the relation between $\ell_{k}$ and $N_{k+1}$.

Begin with $\ell_{0}=5$, and let
\begin{eqnarray*}
A_{0} & = & \{00000,01000\}\\
B_{0} & = & \{11111,10111\}.
\end{eqnarray*}

Next, given $A_{k-1},B_{k-1}$ and $\ell_{k-1}$ and the parameter
$N_{k}$, let $c_{k}$ be a block containing every block in $(A_{k-1}\cup B_{k-1})^{2^{\ell_{k-1}}+1}$,
e.g. enumerate all these blocks and concatenate them.

We proceed in one of two ways, depending on whether $k$ is odd or
even. We denote by $ab$ the concatenation of blocks $a,b$ of symbols,
and by $a^{k}$ the $k$-fold concatenation of a block $a$.
\begin{itemize}
\item If $k$ is odd, let
\begin{eqnarray*}
A_{k} & = & \{c_{k}a^{N_{k}}\,:\, a\in A_{k-1}\}\\
B_{k} & = & \{c_{k}b_{1}b_{2}\ldots b_{N_{k}}\,:\, b_{i}\in B_{k-1}\}.
\end{eqnarray*}

\item If $k$ is even, set
\begin{eqnarray*}
A_{k} & = & \{c_{k}a_{1}a_{2}\ldots a_{N_{k}}\,:\, a_{i}\in A_{k-1}\}\\
B_{k} & = & \{c_{k}b^{N_{k}}\,:\, b\in B_{k-1}\}.
\end{eqnarray*}

\end{itemize}
Thus $A_{k},B_{k}$ consist of blocks of the same length, which we
denote $\ell_{k}$. Note that $\ell_{k}$ can be made arbitrarily
large by increasing $N_{k}$.

If we assume that $N_{k}$ is large enough then one can identify the
occurrences of $c_{k}$ in any long enough subword of length $2\ell_{k}$
of a concatenation of blocks from $A_{k}\cup B_{k}$. This is shown
by induction: first one shows that one can identify the $A_{k-1}\cup B_{k-1}$-blocks,
and then $c_{k}$ is identifiable because it contains blocks from
both $A_{k-1}$ and $B_{k-1}$.

For a set $\Sigma$ let $\Sigma^{*}$ denote the set of all concatenations
of elements from a set $\Sigma$. Given a finite set $L\subseteq\{0,1\}^{*}$
let \[
\left\langle L\right\rangle =\bigcup_{n}T^{n}(L^{\mathbb{N}})\]
denote the subshift consisting of all shifts of concatenations of
blocks from $L$. Note that if $L'\subseteq L^{*}$, then $\left\langle L'\right\rangle \subseteq\left\langle L\right\rangle $.
Let
\[
L_{k}=A_{k}\cup B_{k},
\]
so that $L_{k+1}\subseteq L_{k}^{*}$, and define
\[
X=\bigcap_{k=1}^{\infty}\left\langle L_{k}\right\rangle.
\]
Alternatively, $X$ is the set of points $x\in\{0,1\}^{\mathbb{N}}$
such that every finite block in $x$ appears as a sub-block in a block
from some $L_{k}$.

\section{\label{sec:Analysis}Analysis of the zero-temperature limit}

We make some preliminary observations. For $u\in L_{k}$ let
\[
f_{i}(u)=\mbox{frequency of }i\mbox{ in }u.
\]
Then the following is clear from the construction:
\begin{lem}
\label{lem:frequency-decay}If $N_{k}/\ell_{k-1}$ increases rapidly
enough, then $f_{0}(u)>\frac{2}{3}$ for $u\in A_{k}$ and $f_{0}(u)<\frac{1}{3}$
for $u\in B_{k}$.
\end{lem}
In fact it can be shown that $X$ supports two ergodic measures, respectively
giving mass $>\frac{2}{3}$ and $<\frac{1}{3}$ to the cylinder $[0]$.

The construction is designed so that the ratio $|A_{k}|/|B_{k}|$
fluctuates between very large and very small. More precisely, one
may verify the following:
\begin{lem}
\label{lem:block-counts}If $N_{k}/\ell_{k-1}$ is sufficiently large, then
\begin{itemize}
\item If $k$ is odd then $|B_{k}|>|A_{k}|^{100}$.
\item If $k$ is even then $|A_{k}|>|B_{k}|^{100}$.
\end{itemize}
\end{lem}
The next two lemmas show that for certain values of $\beta$ the measure
$\mu_{\beta\varphi}$ concentrates mostly on blocks from $L_{k}$. Let
\[
Y_{k}=\{x\in\{0,1\}^{\mathbb{N}}\,:\, x|_{[i,i+\ell_{k}-1]}\in L_{k}\mbox{ for some }i\in[0,\ell_{k}-1]\}.
\]
$Y_{k}\subseteq\{0,1\}^{\mathbb{N}}$ is an open and closed set.
\begin{lem}
\label{lem:concentration-on-Lk}For $\beta=2^{3\ell_{k}}$,
\[
\mu_{\beta\varphi}(Y_{k})>1-2^{-\ell_{k}}.
\]
\end{lem}
\begin{proof}
If $x\notin Y_{k}$, then we certainly have $d(x,X)>2^{-2\ell_{k}}$.
Therefore,
\begin{eqnarray*}
\int\beta\varphi d\mu_{\beta\varphi} & = & \int-\beta d(y,X)d\mu_{\beta\varphi}(y)\\
 & < & -2^{3\ell_{k}}\cdot2^{-2\ell_{k}}\mu_{\beta\varphi}(\{0,1\}^{\mathbb{N}}\setminus Y_{k})\\
 & = & 2^{\ell_{k}}(\mu_{\beta\varphi}(Y_{k})-1).
 \end{eqnarray*}
Since $h(\mu_{\beta\varphi})\leq1$ we have
\[
P_{\beta}(\mu_{\beta\varphi})\leq\int\beta\varphi d\mu_{\beta\varphi}+1\leq2^{\ell_{k}}\mu_{\beta\varphi}(Y_{k})-(2^{\ell_{k}}-1).
\]
Finally, choosing $\nu$ to be an invariant measure supported on $X$
we have $P_{\beta}(\nu)=h(\nu)\geq0$, hence $P_{\beta}(\mu_{\beta\varphi})\geq P_{\beta}(\nu)\geq0$.
Combining these we have the desired inequality.\end{proof}
\begin{lem}
\label{lem:decomposition-into-blocks}For $\beta=2^{3\ell_{k}}$,
for all large enough $n$ at least half of the mass of $\mu_{\beta\varphi}$
is concentrated on sequences $u\in\{0,1\}^{n}$ which can be decomposed
as \begin{equation}
u=\diamond v_{1}\diamond\ldots v_{2}\diamond\ldots v_{m}\diamond\label{eq:block-decomposition}\end{equation}
where $v_{i}\in L_{k}$, the symbol $\diamond$ represents blocks of
$0,1$'s (which may vary from place to place), and at least a $(1-2^{-\ell_{k}})$-fraction
of indices $j\in[0,n)$ lie in one of the $v_{i}$. \end{lem}
\begin{proof}
Let
\[
Y'_{k}=\{x\in\{0,1\}^{\mathbb{N}}\,:\, x|_{[0,\ell_{k}-1]}\in L_{k}\}.
\]
Since $Y_{k}=\bigcup_{i=0}^{\ell_{k}-1}T^{-i}Y'_{k}$, from the previous
lemma and shift-invariance of $\mu_{\beta\varphi}$, we see that
\[
\mu_{\beta\varphi}(Y'_{k})>\frac{1}{\ell_{k}}(1-2^{-\ell_{k}}).
\]
Since $\mu_{\beta\varphi}$ is ergodic (being a Gibbs measure), by the ergodic
theorem, for $n$ large enough at least half the mass of $\mu_{\beta\varphi}$
is concentrated on points $x\in X$ such that
\[
\frac{1}{n}\#\{i\in[0,n-1]\,:\, T^{i}x\in Y'_{k}\}>\frac{1}{\ell_{k}}(1-2^{-\ell_{k}}).
\]
Since the beginning of an $L_{k}$-block is uniquely determined (because
the $c_{k}$ blocks can be identified uniquely) we also have that
if $y\in Y'_{k}$, then $T^{i}y\notin Y'_{k}$ for all $1\leq i<\ell_{k}$.
Thus if $u$ is the initial $n$-segment of a point $x$ as above,
then there is a representation of $u$ of the desired form. 
\end{proof}
Next, we obtain a lower bound on $P_{\beta}(\mu_{\beta\varphi})$:
\begin{lem}
\label{lem:Markov-approx-to-Gibbs}If $k$ is odd and $\beta=2^{-3\ell_{k}}$
then
\[
P_{\beta}(\mu_{\beta\varphi})>\frac{\log|B_{k}|}{\ell_{k}}-2^{3\ell_{k}}2^{-\ell_{k}2^{\ell_{k}}}.
\]
A similar statement holds for even $k$ and $A_{k}$.\end{lem}
\begin{proof}
Let $\nu$ be the entropy-maximizing measure on $\left\langle B_{k}\right\rangle $.
Since \[
P_{\beta}(\mu_{\beta\varphi})\geq P_{\beta}(\nu)=h(\nu)-\int\beta\varphi d\nu\]
and $h(\nu)=\frac{\log|B_{k}|}{\ell_{k}}$ it suffices to show that
\begin{equation}
\int\beta\varphi(y)d\nu(y)>-2^{3\ell_{k}}2^{-\ell_{k}2^{\ell_{k}}}.
\label{eq:lower-bound-on-beta-phi-mean}\end{equation}
Indeed, if $y\in\left\langle B_{k}\right\rangle $ then $y=ab_{1}b_{2}\ldots$
where $b_{i}\in B_{k}$ and $a$ is the tail segment of a block in
$B_{k}$. Since, by construction, every concatenation of $2^{\ell_{k}}+1$
blocks from $B_{k}$ appears in $X$, it follows that the initial
segment of $y$ of length $\ell_{k}2^{\ell_{k}}$ appears in $X$,
and therefore $d(y,X)<2^{-\ell_{k}2^{\ell_{k}}}$, and \eqref{eq:lower-bound-on-beta-phi-mean}
follows. 
\end{proof}
The last component of the proof is to show that, for $\beta=2^{3\ell_{k}}$,
the measures $\mu_{\beta\varphi}$ concentrate alternately $B_{k}$ and $A_{k}$.
This is essentially due to the fact that by the lemmas above, $\mu_{\beta\varphi}$
is mostly supported on the blocks of $L_{k}$, and because of the
appearance of entropy in the variational formula, it tends to give
approximately equal mass to these blocks. Since $|B_{k}|/|L_{k}|\rightarrow1$
along the odd integers and $|A_{k}|/|L_{k}|\rightarrow1$ along the
even ones, this implies that $\mu_{\beta\varphi}$ will alternately be supported
mostly on $B_{k}$ and $A_{k}$.

Here are the details. Denote by $[u]$ the cylinder set defined by
a block $u\in\{0,1\}^{*}$.
\begin{prop}
\label{concentration-on-Ak-or-Bk}
If $N_{k}$ increases sufficiently
rapidly then for all $\delta>0$ and all sufficiently large $k$,
if we set $\beta_k=2^{-3\ell_{k}}$ then: if $k$ is odd then
\begin{equation}
\mu_{\beta_k\varphi}\Big(\bigcup_{u\in B_{k}}[u]\Big)\geq1-\delta\label{eq:concentration-on-Bk}
\end{equation}
and if $k$ even then
\[
\mu_{\beta_k\varphi}\Big(\bigcup_{u\in A_{k}}[u]\Big)\geq1-\delta.
\]
\end{prop}
\begin{proof}
We assume that $N_{k}$ increases rapidly enough for the previous
lemmas to hold and furthermore that, writing $H(t)=-t\log t-(1-t)\log(1-t)$,\[
\frac{H(2^{-\ell_{k}})}{\log|B_{k}|/\ell_{k}}\rightarrow0\]
and\[
\frac{2^{3\ell_{k}}2^{-\ell_{k}2^{\ell_{k}}}}{\log|B_{k}|/\ell_{k}}\rightarrow0\]
as $k\rightarrow\infty$ along the odd integers, and similarly, with
$A_{k}$ in place of $B_{k}$, as $k\rightarrow\infty$ along the
even integers. This condition is easily satisfied by choosing $N_{k}$
large enough at each stage, since for fixed $k$, as we increase $N_{k}$
the numerator decays to $0$ but the denominator does not. 

Under these hypotheses we establish the proposition for odd $k$,
the case of even $k$ being similar. Thus, we assume that $|B_{k}|>|A_{k}|^{100}$.
Fix $\delta>0$ and suppose that \eqref{eq:concentration-on-Bk} fails for some $k$.
For all large enough $n$ lemma \ref{lem:decomposition-into-blocks}
implies that at least half the mass of $\mu_{\beta_k\varphi}$ is concentrated
on points whose initial $n$-segment is of the form \eqref{eq:block-decomposition},
and, by the ergodic theorem and the assumed failure of \eqref{eq:concentration-on-Bk},
if $n$ is large then with $\mu_{\beta_k\varphi}$-probability approaching
$1$ the fraction of $v_{i}$'s that belong to $A_{k}$ in the decomposition
\eqref{eq:block-decomposition} is at least $\delta$. 

For such an $n$ we now perform a standard estimate to bound the entropy
of $\mu_{\beta_k\varphi}$. Applying e.g. Stirling's formula, the number of
different ways the $\diamond$'s can appear in $u$ is
\[
\leq\sum_{r<2^{-\ell_{k}}\cdot n}\binom{n}{r}\leq2^{H(2^{-\ell_{k}})n}.
\]
The positions of $\diamond$'s determines the positions of the $v_{i}$,
and given this, the number of ways to fill in the $v_{i}$ so that
at least a $\delta$-fraction of them come from $A_{k}$ is bounded
from above by
\[
\sum_{r=\delta n/\ell_{k}}^{n/\ell_{k}}|A_{k}|^{r}|B_{k}|^{n/\ell_{k}-r}\leq\frac{n}{\ell_{k}}\cdot|A_{k}|^{\delta n/\ell_{k}}|B_{k}|^{(1-\delta)n/\ell_{k}}.
\]
Using the bound $|A_{k}|\leq|B_{k}|^{1/100}$ and setting \[
\delta'=\delta\cdot\frac{99}{100}\]
we get
\[
\leq\frac{n}{\ell_{k}}\cdot|B_{k}|^{(1-\delta')n/\ell_{k}}.
\]
Thus, for arbitrarily large $n$, half the mass of $\mu_{\beta_k\varphi}$
is concentrated on a set $E_{k}\subseteq\{0,1\}^{n}$ of cardinality
\[
|E_{k}|\leq2^{nH(2^{-\ell_{k}})+\log n-\log\ell_{k}}\cdot2^{(1-\delta')n\log|B_{k}|/\ell_{k}}.
\]
It follows from this and the Shannon-McMillan theorem that
\[
h(\mu_{\beta_k\varphi})\leq(1-\delta')\frac{\log|B_{k}|}{\ell_{k}}+H(2^{-\ell_{k}}),
\]
hence, since $\varphi\leq0$, we have
\[
P_{\beta_k}(\mu_{\beta_k\varphi})\leq h(\mu_{\beta_k\varphi})\leq(1-\delta')\frac{\log|B_{k}|}{\ell_{k}}+H(2^{-\ell_{k}}).
\]

Substituting the lower bound from lemma \eqref{lem:Markov-approx-to-Gibbs},
we have
\[
\frac{\log|B_{k}|}{\ell_{k}}-2^{3\ell_{k}}2^{-\ell_{k-1}2^{\ell_{k}}}<
(1-\delta')\frac{\log|B_{k}|}{\ell_{k}}+H(2^{-\ell_{k}}).
\]
By our assumptions about the growth of $N_{k}$ the inequality above
is possible only for finitely many $k$. This completes the proof.
\end{proof}
We can now prove theorem \ref{thm:main-theorem}. For $\delta=\frac{1}{100}$
choose the sequence $N_{k}$ so that the conclusion of the last proposition
holds. Since the density of $0$'s in the blocks $a\in A_{k}$ is
$>\frac{2}{3}$ and the density in the blocks $b\in B_{k}$ is $<\frac{1}{3}$,
it follows that for $k$ large enough and $\beta_{k}=2^{-3\ell_{k}}$,
\begin{eqnarray*}
\mu_{\beta_{k}}([0]) & < & \frac{1}{3}-\delta\qquad\mbox{ if }k\mbox{ is odd}\\
\mu_{\beta_{k}}([0]) & > & \frac{2}{3}+\delta\qquad\mbox{ if }k\mbox{ is even}\end{eqnarray*}
Hence $(\mu_{\beta\varphi})_{\beta\geq0}$ does not weak-{*} converge.

\section{\label{sec:Remarks}Remarks}

\subsection{Topological dynamics of $X$}

In our example $X$ is minimal. Indeed, any block $a\in L_{k}$ appears
in $c_{k+1}$ and hence in every block in $L_{k+1}$, so $a$ appears
in $X$ with bounded gaps. Note that there are also minimal (non uniquely
ergodic) systems $X$ for which the zero-temperature limit exists.

One can easily modify the construction to endow $X$ with other dynamical
properties, e.g. one can make $X$ topologically mixing (our example
is not, in fact it has a periodic factor of order $5$). It is also simple
to obtain positive entropy of $X$ (and the
limiting measures): form the product of the given example with a full
shift.

\subsection{Measurable dynamics of the zero-temperature limits}

In our example, $(\mu_{\beta\varphi})_{\beta\geq0}$ has two ergodic accumulation
points, and one can show that the convex combinations of these two
are also accumulation points.

In general, the set of accumulation points need not contain ergodic
measures, even when the zero-temperature limit exists. This is true
even of locally constant potentials \cite{Leplaideur05,ChazottesGambaudoUgalde09},
and one can also construct examples which are simpler to analyze.
For example, if $X\subseteq\{0,1\}^{\mathbb{N}}$ is a subshift invariant
under involution $0\leftrightarrow1$ of $\{0,1\}^{\mathbb{N}}$,
and if $X$ has precisely two invariant measures $\mu',\mu''$ which
are exchanged by this involution, then for the potential $\varphi_{X}(y)=-d(y,X)$
we will have $\lim_{\beta\rightarrow+\infty}\mu_{\beta\varphi}=\frac{1}{2}\mu'+\frac{1}{2}\mu''$. 

The set of accumulation points also need not be convex. Using the
same scheme as above one can construct a subshift $X\subseteq\{1,2,3\}^{\mathbb{Z}}$
with three invariant measures $\mu^{(i)},i=1,2,3$, by maintaining
three sets of blocks $A_{k},B_{k},C_{k}$ at each stage (rather than
two). At each step of the construction we choose the smallest of the
sets and concatenate its blocks freely, but concatenate the blocks
of the others in a constrained way, so that at the next stage the
sizes of the selected set is much larger than the other two, which
have not changed much in relative size. For each $n$ there are always
two sets (the two which are not growing very much at that stage) for
which the number of $n$-blocks in one is much greater than in the
other. Thus the Gibbs measures at the appropriate scale will have
very small contributions from the smaller of these sets, and the accumulation
points of $\mu_{\beta\varphi}$ will lie near the boundary of the simplex
spanned by the $\mu^{(i)}$ (in our example there were only two sets
and at each step one grew at the expense of the other; thus the relative
number of $n$- blocks achieved all intermediate ratios). 

Regarding the ergodic nature of the accumulation points, the same
periodicity of order five that obstructs topological mixing causes
the ergodic invariant measures on $X$ (i.e. the ergodic zero-temperature
limits) to have $e^{-2\pi i/5}$ in their spectrum, but this can be
avoided by introducing spacers into the construction. In this way
one can make the limiting ergodic measures weak or strong mixing,
and possibly $K$. 

Finally, we have the following variant of Radin's argument from \cite{Radin91}.
Let $\mu$ be an ergodic probability measure for some measurable transformation
of a Borel space, and $h(\mu)<\infty$. By the Jewett-Krieger theorem
\cite{DenkerGrillenbergerSigmund76} there is a subshift $X$ on at
most $h(\mu)+1$ symbols whose unique shift-invariant measure $\nu$
is isomorphic to $\mu$ in the ergodic theory sense. For the potential
$\varphi_{X}$, all accumulation points of $\mu_{\beta\varphi}$ are invariant
measures on $X$, so they all equal $\nu$; thus $\mu_{\beta\varphi}\rightarrow\nu$
as $\beta\rightarrow+\infty$. This shows that the zero-temperature
limit of Gibbs measures can have arbitrary isomorphism type, subject
to the finite entropy constraint, and raises the analogous question
for divergent potentials: 
\begin{problem*}
Given arbitrary ergodic measures $\mu',\mu''$ of the same finite
entropy, can one construct a H\"{o}lder potential $\varphi$ whose
Gibbs measures $\mu_{\beta\varphi}$ have two ergodic accumulation points
as $\beta\rightarrow+\infty$, isomorphic respectively to $\mu',\mu''$? 
\end{problem*}

\subsection{Maximization of marginal entropy}

Let $\varphi$ be a H\"{o}lder potential and $\mathcal{M}$ the set
of invariant probability measures $\mu$ for which $\int\varphi d\mu$
is maximal. It is known that if $\mu$ is an accumulation point of
$(\mu_{\beta\varphi})_{\beta>0}$ then $\mu\in\mathcal{M}$ and furthermore
$\mu$ maximizes $h(\mu)$ subject to this condition. 

In the example constructed above the potential $\varphi$ had two
$\varphi$-maximizing ergodic measures $\mu',\mu''$, and the key
property that we utilized was that their marginals at certain scales
had sufficiently different entropies. In fact, the measure maximizing
the marginal entropy on $\{0,1\}^{n}$ for certain $n$ was alternately
very close to $\mu'$ and to $\mu''$. 

It is an interesting question if such a connection between zero-temperature
convergence and marginal entropy exists in general. Let $\varphi$
be a H\"{o}lder potential, and for each $n$ let $\mathcal{M}_{n}^{*}$
denote the set of marginal distributions produced by restricting $\mu\in\mathcal{M}$
to $\{0,1\}^{n}$. The entropy function $H(\cdot)$ is strictly concave
on $\mathcal{M}_{n}^{*}$, and therefore there is a unique $\mu_{n}^{*}\in\mathcal{M}_{n}^{*}$
maximizing the entropy function. Let
\[
\mathcal{M}_{n}=\{\mu\in\mathcal{M}\,:\,\mu|_{\{0,1\}^{n}}=\mu_{n}^{*}\}.
\]
This is the set of $\varphi$-maximizing measures which maximize entropy
on $n$-blocks. Note that the diameter of $\mathcal{M}_{n}$ tends
to $0$ as $n\rightarrow\infty$ in any weak-{*} compatible metric.
Hence we can interpret $\mathcal{M}_{n}\rightarrow\mu$ in the obvious
way. 
\begin{problem*}
Is the existence of a zero-temperature limit for $\varphi$ equivalent
to existence of $\lim\mathcal{M}_{n}$? More generally, do $(\mu_{\beta\varphi})_{\beta\geq0}$
and $(\mathcal{M}_{n})_{n\geq0}$ have the same accumulation points?
\end{problem*}

\section{\label{sec:The-multidimensional-case}The multidimensional case}

In this section we apply the main theorem of \cite{Hochman09} to
obtain a locally constant potential (i.e. a finite-range interaction) in dimension $d\geq3$ such that
any associated family of equilibrium measures does not converge as $\beta\rightarrow+\infty$.
Our methods do not work in $d=2$, because the results of \cite{Hochman09} are not known in that case,
but probably a more direct construction is possible.

\subsection{\label{sub:SFTs-and-subdynamics}SFTs and their subdynamics }

The metric on $\{0,1\}^{\mathbb{Z}^{d}}$ is defined by%
\footnote{The dimension of the ambient space is also denoted $d$ but no confusion
should arise.%
}\[
d(x,y)=2^{-\min\{\left\Vert u\right\Vert \,:\, x(u)\neq y(u)\}}\]
where $\left\Vert \cdot\right\Vert $ is the sup-norm. We denote by
$T$ the shift action on $\{0,1\}^{\mathbb{Z}^{d}}$ and write $T_{1},\ldots,T_{d}$
for its generators.

Let \[
E_{n}=\{-n,\ldots,0,,\ldots,n\}^{d}\]
denote the discrete $d$-dimensional cube of side $2n+1$. A subshift
$X$ is a shift of finite type (SFT) if there is an $n$ and finite
set of patterns $L\subseteq\{0,1\}^{E_{n}}$ such that
\[
X=\{x\in\{0,1\}^{\mathbb{Z}^{d}}\,:\,\mbox{no pattern from }L\mbox{ appears in }x\}.
\]
(Note: here $L$ determines the forbidden patterns, which is the opposite
of its role in $\left\langle L\right\rangle$.) A pattern $a$ is
said to be \emph{locally admissible }if it does not contain any patterns
from $L$; it is \emph{globally admissible } if it appears in $X$,
i.e. it can be extended to a configuration on all of $\mathbb{Z}^{d}$
which does not contain patterns from $L$. These two notions are distinct,
and it is formally impossible to decide in general, given $L$, whether
a locally admissible word is globally admissible.

If we write \begin{equation}
\varphi_{L}(y)=\left\{ \begin{array}{cc}
-1 & y|_{E_{n}}\in L\\
0 & \mbox{otherwise}\end{array}\right.\label{eq:SFT-potential}\end{equation}
then every invariant measure $\mu$ on $\{0,1\}^{\mathbb{Z}^{k}}$
satisfies $\int\varphi_{L}d\mu\leq0$ with equality if and only if
$\mu$ is supported on $X$. Thus for any SFT $X$ there is a locally
constant potential whose maximizing measures are precisely the invariant
measures on $X$.

Given a subshift $X\subseteq\{0,1\}^{\mathbb{Z}^{d}}$, we may consider
the restricted one-parameter action of $T_{1}$ on $X$. We shall
say that the topological dynamical system $(X,T_{1})$ is a (one-dimensional)
subaction of $(X,T)$. To each partition $C=\{C_{1},\ldots,C_{m}\}$
of $\{0,1\}^{\mathbb{Z}^{d}}$ into closed and open sets we associate
to each $x\in X$ its itinerary $x^{C}$  given by the action of $T_{1}$
and the partition $C$, i.e. $x\mapsto x^{C}\in\{1,\ldots,m\}^{\mathbb{Z}}$
is defined by
\[
x^{C}(i)=j\mbox{ if and only if }T_{1}^{i}x\in C_{j}.
\]
The subshift
\[
X^{C}=\{x^{C}\,:\, x\in X\}\subseteq\{1,\ldots,m\}^{\mathbb{Z}}
\]
is a factor, in the sense of topological dynamics, of the subaction
$(X,T_{1})$.

For a subshift $Y\subseteq\{0,1\}^{\mathbb{Z}}$ write $L_{k}(Y)\subseteq\{0,1\}^{k}$
for the set of $k$-blocks appearing in $Y$; note that for any sequence
$k(i)\rightarrow\infty$ the sets $L_{k(i)}$, $i=1,2,\ldots$, determine $Y$.

The main result of \cite{Hochman09} says that the subaction of SFTs
can be made to look like an arbitrary subshift, as long as the subshift
is constructive in a certain formal sense. The version we need is
the following:
\begin{thm*}
\label{thm:subdynamics}Let $A$ be an algorithm that for each $i$
computes%
\footnote{A stronger statement can be made in which the computability is replaced
with semi-computability of an appropriate family of blocks, and then
one obtains (nearly) a characterization; but we do not need this here.%
} an integer $n(i)$ and a set $L_{i}\subseteq\{0,1,\ldots,r\}^{n(i)}$
such that $\left\langle L_{i}\right\rangle \supseteq\left\langle L_{i+1}\right\rangle $.
Then there is an alphabet $\Sigma$, an SFT $X\subseteq\Sigma^{\mathbb{Z}^{3}}$
of entropy $0$ and a closed and open partition $C=\{C_{0},C_{1},\ldots,C_{r}\}$
of $\Sigma^{\mathbb{Z}^{3}}$ such that $L_{n(i)}(X^{C})=L_{i}$,
and consequently $X^{C}=\cap\left\langle L_{i}\right\rangle $. Furthermore,
the partition elements $C_{i}$ can be made invariant under the shifts
$T_{2}$ and $T_{3}$.
\end{thm*}
To apply this one usually begins with a subshift $Y$ which has
been constructed in some explicit manner, and a computable sequence
$n(i)$ (e.g. $n(i)=i$), and derives an algorithm which from $i$
computes $L_{n(i)}(Y)$; one then gets an SFT $X$ and partition $C$
so that $X^{C}=Y$. This means that for all practical nearly purposes (e.g.
the construction of counterexamples) one can realize arbitrary dynamics
as the subdynamics of an SFT.\footnote{Nevertheless, one should bear in mind that the family of SFTs (and the set of algorithms) is countable.}

From the result for dimension $d=3$ the same is easily seen to hold
for $d\geq3$, but it is not known whether this holds in dimension
$d=2$.

\subsection{\label{sub:A-modified-one-dim-construction}A modified one-dimensional
example}

For notational convenience, for the rest of the paper we concentrate
on the case $d=3$, the general case being similar.

Realizing a specific subshift (such as the one from section \ref{sec:Construction-of-X})
as the subaction of an SFT $X$ does not in itself give good control
over the equilibrium measures of $\varphi_{X}$ or $\varphi_{L}$. Indeed,
the size of $L_{n}(X^{C})$ is exponential in $n$, which implies
similar growth of the corresponding set $L_{n}(X)$, but does not
guarantee exponential growth in $n^{3}$, which is the appropriate
scale for 3-dimensional subshifts. Thus for example we can have $h(X^{C})>0$
but $h(X)=0$. 

In order to use subactions to control entropy of the full $\mathbb{Z}^{3}$
action we rely on a trick by which the frequency of symbols in $X^{C}$
can be used to control pattern counts in a certain extension of $X$.
This approach was used in \cite{HochmanMeyerovitch07,BoyleSchraudner09}. 

We begin by modifying the main example of this paper so as to control
frequencies rather than block counts. We define a sequence of integers
$\ell_{k}$ and sets of blocks $A_{k},B_{k}\subseteq\{0,1,2\}^{\ell_{k}}$
by induction, using an auxiliary sequence $N_{1},N_{2}\ldots$ of
integers.

Start with $\ell_{0}=2$ and $A_{0}=\{00,01\}$, $B_{0}=\{00,02\}$.
Next, given $k$ define
\begin{eqnarray*}
A_{k} & = & \{a^{1+2^{\ell_{k-1}}+N_{k}}\,:\, a\in A_{k-1}\}\\
B_{k} & = & \{b^{1+2^{\ell_{k-1}}}2^{N_{k}\ell_{k-1}}\,:\, b\in B_{k-1}\}
\end{eqnarray*}
and for $k$ even define
\begin{eqnarray*}
A_{k} & = & \{a^{1+2^{\ell_{k-1}}}1^{N_{k}\ell_{k-1}}\,:\, a\in A_{k-1}\}\\
B_{k} & = & \{b^{1+2^{\ell_{k-1}}+N_{k}}\,:\, b\in B_{k}\}.
\end{eqnarray*}
Let $\ell_{k}$ be the common length of blocks in the sets above,
i.e. $\ell_{k}=\ell_{k-1}(2^{\ell_{k-1}+1}+N_{k})$. Note that $1^{\ell_{k}}\in A_{k}$
and $2^{\ell_{k}}\in B_{k}$.

As $k\rightarrow\infty$ the frequency of $0$'s in the blocks of
$A_{k},B_{k}$ tends to $0$, and the frequency of $1$'s and $2's$
tends, respectively, to $1$, and we can control the relative speed
at which they do so. More precisely, there is a function $\widetilde{N}_{k}(\cdot)$
such that given $N_{1},\ldots,N_{k-1}$ and $N_{k}\geq\widetilde{N}_{k}(N_{1},\ldots,N_{k-1})$
we have
\begin{eqnarray*}
f_{0}(a) & > & 100f_{0}(b)\qquad\mbox{for }k\mbox{ odd, }a\in A_{k},b\in B_{k}\\
f_{0}(b) & > & 100f_{0}(a)\qquad\mbox{for }k\mbox{ even, }a\in A_{k},b\in B_{k}.
\end{eqnarray*}
(Recall that $f_{0}(x)$ is the frequency of the symbol $0$ in $x$.)

Define
\[
Y=\bigcap_{k=1}^{\infty}\left\langle A_{k}\cup B_{k}\right\rangle.
\]
Similarly define
\[
Y_{1}=\bigcap\left\langle A_{k}\right\rangle
\]
and
\[
Y_{2}=\bigcap\left\langle B_{k}\right\rangle.
\]
(notice that these are decreasing intersections). Note that the only invariant
measures on $Y$ are the point masses at the fixed points $1^{\infty}\in Y_{1}$
and $2^{\infty}\in Y_{2}$. We denote \begin{equation}
\ell'_{k}=\ell_{k-1}\left((|A_{K}|+|B_{k}|)^{M_{k}}+\widetilde{N}_{k}(N_{1}\ldots N_{k-1})\right)\label{eq:ell-prime}\end{equation}
(so $\ell_{k}\geq\ell'_{k}$) and note that as long as $N_{k}\geq\widetilde{N}_{k}(N_{1},\ldots,N_{k-1})$,
the set $L_{\ell'_{k}}(Y)$ is in fact independent of $N_{k}$ and
depends only on $N_{1},\ldots,N_{k-1}$. We also note that $\widetilde{N}_{k}$
can be computed explicitly, and in particular the function $(k,N_{1},\ldots,N_{k-1})\mapsto\widetilde{N}_{k}(N_{1},\ldots,N_{k-1})$
is a formally computable function.

\subsection{\label{sub:Controlling-pattern-counts}Controlling pattern counts
in a 3-dimensional SFT}

We now incorporate the subshift $Y$ constructed above into a 3-dimensional
SFT and use the control over the frequency of symbols in $Y$ to gain
control of the pattern counts of an associated SFT.

First, some notation: for a subshift $X\subseteq\Sigma^{\mathbb{Z}^{3}}$
write \[
L_{n}(X)=\{x|_{E_{n}}\,:\, x\in X\}\subseteq\Sigma^{E_{n}}\]
where $E_{n}=\{-n,\ldots,n\}^{3}$. This is the same notation we used
for one-dimensional subshifts, but the meaning will be clear from
the context. We remark that if the (topological) entropy of $X$ is
$0$ then $|L_{n}(X)|=o(|E_{n}|)$.

Apply theorem \ref{thm:subdynamics} to $Y$ (or, rather, to an algorithm
that computes a sequence $L_{n(k)}(Y)$; we shall be more precise
later about the algorithm used). We obtain a zero-entropy SFT $X\subseteq\Sigma^{\mathbb{Z}^{3}}$
and $C=\{C_{0},C_{1},C_{2}\}$ a $T_{2},T_{3}$-invariant partition
so that $X^{C}=Y$. 

Next, for $x\in X$ and $u=(u_{1},u_{2},u_{3})\in\mathbb{Z}^{3}$,
if $x^{C}(u_{1})=0$ (i.e. if $T_{1}^{u_{1}}x\in C_{0})$ we {}``color''
the site with one of the two colors $0',0''$. Otherwise we leave
it {}``blank''. Collect all such colorings into a new subshift $\widehat{X}$.
Formally, $\widehat{X}\subseteq X\times\{0',0'',\mbox{blank}\}^{\mathbb{Z}^{3}}$
is defined by
\[
\widehat{X}=\{(x,y)\in X\times\{0',0'',\mbox{blank}\}\,:\, y(u)=\mbox{blank if }x^{C}(u_{1})\neq0\}.
\]
For $x=(x_{1},x_{2})\in\widehat{X}$ we also write $x^{C}$ instead
of $x_{1}^{C}$. One may verify that $\widehat{X}$ is an SFT. We
write $\widehat{\Sigma}=\Sigma\times\{0',0'',\mbox{blank}\}$ for
the alphabet of $\widehat{X}$ and write $\widehat{L}$ for the finite
set of patterns whose exclusion defines $\widehat{X}$. We may assume
that if a pattern over $\widehat{\Sigma}$ is locally admissible for
$\widehat{L}$ then the pattern induced from its first component is
locally admissible for $L$. 

Notice that, since $C_{0},C_{1},C_{2}$ are invariant under $T_{2},T_{3}$,
the pattern of symbols $0',0''$ in a point $x\in\widehat{X}$ is
the union of affine planes whose direction is spanned by $(0,1,0),(0,0,1)$.
The sequence of coordinates at which these planes intersect the $x$-axis
corresponds to the location of $0$-s in $x^{C}$, and on each plane
the symbols $0',0''$ are distributed as randomly as possible, i.e.
given the arrangement of affine planes there is no restriction on
the combinations of $0',0''$ that may appear in them. It follows
that if $a\in\{0,1,2\}^{\{-n,\ldots,n\}}$ is a block in $Y$ then
\[
\#\{(x,y)|_{E_{n}}\,:\,(x,y)\in\widehat{X}\mbox{ and }x^{C}|_{\{-n,\ldots,n\}}=a\}=2^{f_{0}(a)|E_{n}|+o(|E_{n}|)}.
\]
(the term $o(|E_{n}|)$ comes from the pattern growth of $X$, which has entropy $0$). 

Write
\begin{eqnarray*}
\widehat{X}_{1} & = & \{x\in\widehat{X}\,:\, x^{C}\in Y_{1}\}\\
\widehat{X}_{2} & = & \{x\in\widehat{X}\,:\, x^{C}\in Y_{2}\}.
\end{eqnarray*}
Then, for $k$ large enough, the frequency gap between blocks in $A_{k}$
and $B_{k}$ translates into
\begin{eqnarray*}
|L_{\ell_{k}}(\widehat{X}_{1})| & > & |L_{\ell_{k}}(\widehat{X}_{2})|^{1/10}\qquad k\mbox{ odd}\\
|L_{\ell_{k}}(\widehat{X}_{2})| & > & |L_{\ell_{k}}(\widehat{X}_{1})|^{1/10}\qquad k\mbox{ even}.\end{eqnarray*}
Compare this with lemma \ref{lem:block-counts}.

\subsection{\label{sub:Local-versus-global-admissibility}Local versus global
admissibility}

For $\varphi=\varphi_{\widehat{X}}$, i.e. $\varphi(y)=-d(y,\widehat{X})$,
one can adapt the analysis in section \ref{sec:Analysis} and show
that $\mu_{\beta\varphi}$ does not have a limit as $\beta\rightarrow+\infty$.
Let us review this argument. Fix $\beta=2^{-3\ell_{k}}$, and set
$p=1,2$ according to whether $k$ is odd or even, and write $q=2-p$
for the other index. First, an in Lemma \ref{lem:Markov-approx-to-Gibbs},
we prove a lower bound on $P_{\beta}(\mu_{\beta\varphi})$ by constructing
a measure $\nu_{k}$ whose blocks (i.e. square patterns) are overwhelmingly
drawn from $L_{\ell_{k}}(X_{p})$, making it nearly $\varphi_{\widehat{X}}$-maximizing,
and with entropy close to $\frac{1}{|E_{\ell_{k}}|}|L_{\ell_{k}}(\widehat{X}_{p})|$.
This forces the entropy of $\mu_{\beta\varphi}$ to be similar. Second, we
use the fact that most of the mass of $\mu_{\beta\varphi}$ concentrates on blocks from
$L_{\ell_{k}}(\widehat{X})$ and the fact that $L_{\ell_{k}}(\widehat{X}_{p})\gg L_{\ell_{k}}(\widehat{X}_{q})$
to deduce that in order for $\mu_{\beta\varphi}$ to have entropy near $\frac{1}{|E_{\ell_{k}}|}|L_{\ell_{k}}(\widehat{X}_{p})|$,
it must be mostly concentrated on $\widehat{X}_{p}$. This argument
is similar to that in proposition \ref{concentration-on-Ak-or-Bk}.

We are now interested in proving the same thing for the potential
$\varphi_{\widehat{L}}$ (given in \eqref{eq:SFT-potential}) instead
of $\varphi_{\widehat{X}}$. The first part of the analysis above
carries over with only minor modifications.

However, the second part runs into difficulties. Notice that $\int\varphi_{\widehat{X}}d\mu\approx0$
implies that nearly all the $\mu_{\beta\varphi}$-mass is concentrated on
patterns in $L_{\ell_{k}}(\widehat{X})$, but $\int\varphi_{\widehat{L}}d\mu\approx0$
tells us only that $\mu_{\beta\varphi}$-most blocks on $E_{\ell_{k}}$ are
\emph{locally }admissible for $\widehat{L}$; they do not have to
be \emph{globally }admissible, giving us little control of their structure. 

To pull things through, we will make use of the following observation:
it is not necessary for us to know that most of the mass of $\mu_{\beta\varphi}$
concentrates on $L_{\ell_{k}}(\widehat{X})$. Instead, it suffices
that it concentrates on $L_{\ell'_{k}}(\widehat{X})$, where $\ell'$
is as in equation \eqref{eq:ell-prime}. This is because $L_{\ell'_{k}}(\widehat{X}_{p})$
is already much larger than $L_{\ell'_{k}}(\widehat{X}_{q})$, so
we can argue as in the first part of the proof of proposition \ref{concentration-on-Ak-or-Bk}. 

Thus, to complete the construction we want to ensure that if a block
$a\in\Sigma^{E_{\ell_{k}}}$ is locally admissible then $a|_{E_{\ell'_{k}}}$
is globally admissible, i.e. belongs to $L_{\ell'_{k}}(\widehat{X})$. 

A simple compactness argument establishes the following general fact:
For any SFT and $m\in\mathbb{N}$ there is an $R$ so that if $b\in\Sigma^{E_{R}}$
is locally admissible then $b|_{E_{m}}$ is globally admissible. In
general, however, $R$ depends in a very complicated way on both the
SFT and $m$, and in fact is not formally computable given these parameters.
For our purposes we require finer control than this. Luckily, an inspection
of the proof in \cite{Hochman09} gives the following:
\begin{thm}
\label{thm:subdynamics-refined}Let $A$ be an algorithm that from
$i$ computes $n(i)\in\mathbb{N}$ and $L_{i}\subseteq\{0,1,\ldots,r\}^{n(i)}$
such that $\left\langle L_{i}\right\rangle \supseteq\left\langle L_{i+1}\right\rangle $.
Denote by $\tau_{i}$ the number of time-steps required for the computation
on input $i$. Then the SFT $X$ from theorem \ref{thm:subdynamics}
can be chosen so that, for $R_{i}=R_{i}(|A|,\tau_{1},\ldots,\tau_{i})$
, if $a\in\Sigma^{E_{R_{i}}}$ is locally admissible then $a|_{E_{n(i)}}$
is globally admissible, and furthermore the function $R_{i}(\ldots)$
is computable. Here $\tau_{i}$ and $A$ are taken with respect to
some fixed universal Turing machine.
\end{thm}

\subsection{\label{sub:Completing-the-construction}Completing the construction:
The fine print}

We now specify an algorithm $A$ which, given $i$, computes sequences
$n(i)\in\mathbb{N}$ and $L_{i}\subseteq\{0,1,2\}^{n(i)}$ so that
$\left\langle L_{i}\right\rangle \supseteq\left\langle L_{i+1}\right\rangle $.
The even elements $n(2k)$ are the lengths $\ell_{k}$ associated
to a sequence $N_{k}$ in the construction in section \ref{sub:A-modified-one-dim-construction},
i.e.
\[
N_{k}=\frac{n(2k)-n(2k-2)\cdot2^{1+n(2k-2)}}{n(2k-2)}\cdot
\]
The odd elements of the sequence are
\[
n(2k-1)=\ell'_{k}=\ell_{k-1}\widetilde{N}_{k}(N_{1},\ldots,N_{k-1}).
\]
Note that, having determined $n(i)$, the blocks in $Y$ of length
$n(i)$ depend only on $N_{1},\ldots,N_{[n(i)/2]}$ and not on any
future choices of parameters of the construction. Hence $L_{i}=L_{n(i)}(Y)$
is well defined given $n(1),\ldots,n(i)$ and may be computed from
this data. Thus at the $i$-th stage of the computation we will write
$L_{n(i)}(Y)$ even though strictly speaking $Y$ is not yet defined.

On input $i$ the algorithm is as follows. 
\begin{description}
\item [{Case~0}] $i=1$. Output
\begin{eqnarray*}
n(1) & = & 1\\
L_{1} & = & \{0,1,2\}.
\end{eqnarray*}

\item [{Case~1}] $i=2k-1$. Recursively compute $N_{1},\ldots,N_{k-1}$,
and output
\begin{eqnarray*}
n(i) & = & \ell'_{k}\qquad=\qquad\ell_{k-1}\widetilde{N}_{k}(N_{1},\ldots,N_{k-1})\\
L_{i} & = & L_{n(i)}(Y).
\end{eqnarray*}

\item [{Case~2}] $i=2k$. Recursively compute $N_{m}$,$m<k$ and
the time $\tau_{1},\ldots,\tau_{i-1}$ spent by the algorithm when
run on each of the inputs $j=1,\ldots,i-1$. Let
\begin{eqnarray*}
N_{k} & = & \left(\max\{n(i-1),R(|A|,\tau_{1},\ldots,\tau_{i-1})\}\right)^{2}
\end{eqnarray*}
and output
\begin{eqnarray*}
n(i) & = & \ell_{k}\quad=\quad N_{k}\ell_{k-1}\\
L_{i} & = & L_{n(i)}(Y).
\end{eqnarray*}

\end{description}
Realizing such an algorithm (which can simulate itself) is a non-trivial
but standard exercise in computation theory.

We can now sketch the remainder of the proof of theorem \ref{thm:multidim}.
Using $A$ as input to theorem \ref{thm:subdynamics} we obtain an
SFT $X\subseteq\Sigma^{\mathbb{Z}^{3}}$ and associated partition
$C=\{C_{0},C_{1},C_{2}\}$ of $\{0,1,2\}^{\mathbb{Z}^{3}}$, invariant
under $T_{2},T_{3}$, such that $X^{C}=Y$. Next, form the SFT $\widehat{X}$
as explained above, defined by a set $\widehat{L}$ of excluded patterns. 

For $\beta=2^{3\ell_{k}}$ let $\mu_{\beta\varphi}$ be an equilibrium measure
associated to the potential $\varphi_{\widehat{L}}$. By the definition
of equilibrium measures we have $\int\varphi_{\widehat{L}}d\mu_{\beta\varphi}>-c2^{-3\ell_{k}}$,
where $c=\log|\widehat{\Sigma}|$ is the maximal entropy achieved
by an invariant measures on the full shift $\widehat{\Sigma}^{\mathbb{Z}^{3}}$;
in section \ref{sec:Analysis} this constant was $1$. Thus in a $\mu_{\beta\varphi}$-typical
configuration the density of patterns from $\widehat{L}$ is $<c2^{-3\ell_{k}}$.
Hence for $r=\sqrt{\ell_{k}}$ and large enough $k$, with $\mu_{\beta\varphi}$-probability
$>1-2^{-2\ell_{k}}$ a configuration $x$ satisfies that $x|_{E_{r}}$
is globally admissible. By our choice of $\ell_{k}=n(2k)$ we have
$r\geq R(|A|,N_{1},\ldots,N_{k})$, so $x|_{E_{n(2k-1)}}$\emph{ }is
\emph{globally }admissible. But since $n(2k-1)\geq\ell'_{k}$, we
are in the situation described at the end of the previous subsection,
and this is enough to conclude that $\mu_{\beta\varphi}$ is mostly concentrated
on $\widehat{X}_{1}$ or $\widehat{X}_{2}$, depending on $k\bmod2$;
so $\mu_{\beta\varphi}$ diverges along $\beta=2^{-3\ell_{k}}$. 

\bibliographystyle{plain}
\bibliography{bib-3dec09}

\end{document}